\documentclass[11pt]{article}
\usepackage{pdfpages}
\usepackage{multirow}
\usepackage{array}
\usepackage{graphicx}
\usepackage{subcaption}
\usepackage{latexsym}
\usepackage{pstricks}
\usepackage{pst-node}
\usepackage{pst-tree}
\usepackage{url}
\usepackage{times}
\usepackage{helvet}
\usepackage{courier}
\usepackage{algorithmic}
\usepackage[ruled,vlined,linesnumbered]{algorithm2e}
\usepackage{amsthm}
\usepackage{graphicx,amssymb,amsmath}
\usepackage{mathrsfs}
\usepackage{mathtools}
\usepackage{comment}

\usepackage[hypertexnames=false, colorlinks=true, citecolor=blue, linkcolor=red, urlcolor=black]{hyperref}

\usepackage{geometry}
\usepackage{authblk}

\makeatletter
\renewcommand*{\verbatim@font}{\sffamily}

\title{The Curse and Blessing of Not-All-Equal in k-Satisfiability}

\author{S. Cliff Liu\thanks{Department of Computer Science, Princeton University, Princeton, NJ 08540, USA.
Email: sixuel@princeton.edu.}
}

\begin{document}

\maketitle


\newcounter{dummy} \numberwithin{dummy}{section}
\newtheorem{lemma}[dummy]{Lemma}
\newtheorem{definition}[dummy]{Definition}
\newtheorem{remark}[dummy]{Remark}
\newtheorem{corollary}[dummy]{Corollary}
\newtheorem{claim}[dummy]{Claim}
\newtheorem{observation}[dummy]{Observation}

\newtheorem{theorem}{Theorem}

\renewcommand{\algorithmicrequire}{\textbf{Input:}}
\renewcommand{\algorithmicensure}{\textbf{Output:}}

\begin{abstract}
    As a natural variant of the $k$-SAT problem, NAE-$k$-SAT additionally requires the literals in each clause to take not-all-equal (NAE) truth values. In this paper, we study the worst-case time complexities of solving NAE-$k$-SAT and MAX-NAE-$k$-SAT approximation, as functions of $k$, the number of variables $n$, and the performance ratio $\delta$. The latter problem asks for a solution of at least $\delta$ times the optimal. Our main results include:
    \begin{itemize}
        \item A deterministic algorithm for NAE-$k$-SAT that is faster than the best deterministic algorithm for $k$-SAT on all $k \ge 3$. Previously, no NAE-$k$-SAT algorithm is known to be faster than $k$-SAT algorithms. For $k = 3$, we achieve an upper bound of $1.326^n$. The corresponding bound for $3$-SAT is $1.328^n$.
        \item A randomized algorithm for MAX-NAE-$k$-SAT approximation, with upper bound $(2 - \epsilon_k(\delta))^n$ where $\epsilon_k(\delta) > 0$ only depends on $k$ and $\delta$. Previously, no upper bound better than the trivial $2^n$ is known for MAX-NAE-$k$-SAT approximation on $k \ge 4$. For $\delta = 0.9$ and $k = 4$, we achieve an upper bound of $1.947^n$.
        \item A deterministic algorithm for MAX-NAE-$k$-SAT approximation. For $\delta = 0.9$ and $k = 3$, we achieve an upper bound of $1.698^n$, which is better than the upper bound $1.731^n$ of the exact algorithm for MAX-NAE-$3$-SAT.
    \end{itemize}
    Our finding sheds new light on the following question: Is NAE-$k$-SAT easier than $k$-SAT? The answer might be affirmative at least on solving the problems exactly and deterministically, while approximately solving MAX-NAE-$k$-SAT might be harder than MAX-$k$-SAT on $k \ge 4$.
\end{abstract}

\thispagestyle{empty}

\newpage

\setcounter{page}{1}

\section{Introduction}\label{intro}

Like the well-known $k$-SAT and MAX-$k$-SAT problems,
NAE-$k$-SAT is \textsf{NP}-complete and MAX-NAE-$k$-SAT is \textsf{APX}-complete,
and they have intimate relationships to other well-known combinatorial problems such as MAX-Set-Splitting, Hypergraph-Coloring, MAX-CUT, and MAX-XOR-$k$-SAT \cite{DBLP:journals/siamcomp/GuruswamiHS02}.
Extending the methods in solving $k$-SAT to NAE-$k$-SAT usually leads to fruitful results,
e.g., in the study of polynomial-time approximation algorithms
(see Fig.~\ref{poly_inapx} in Appendix~\ref{subsec_poly_table}).
The symmetry in NAE-$k$-SAT can also be very useful in the context of random satisfiability \cite{DBLP:conf/soda/AchlioptasCIM01, DBLP:conf/stoc/DingSS14, DBLP:conf/focs/SlySZ16}.
However, whether the NAE predicate plays a positive role in the analysis of worst-case upper bounds remains open.

In this paper, we study the exponential upper bounds for solving NAE-$k$-SAT
exactly and approximately (MAX-NAE-$k$-SAT approximation).
The research in this area mainly concentrates on two separated lines: exponential-time exact algorithms and polynomial-time approximations.
Therefore, another reason motivating this work is to combine both lines by designing \emph{moderately} exponential-time algorithms to achieve performance ratio beyond the inapproximable
threshold of polynomial-time algorithms.


Traditionally, problems in \textsf{P} are well-studied.
But due to the widely-believed Exponential-Time Hypothesis (ETH), sub-exponential-time algorithms are unlikely to exist for $k$-SAT, let alone MAX-$k$-SAT and MAX-NAE-$k$-SAT \cite{DBLP:journals/jcss/ImpagliazzoP01}.
So understanding the quality of exponential-time algorithms is important, which essentially tells us which problem is less intractable.
As a strong evidence, lots of conceptual breakthroughs have been made by consecutive advancements of faster exponential-time algorithms for $k$-SAT
\cite{DBLP:conf/focs/Schoning99,DBLP:journals/jacm/PaturiPSZ05,DBLP:conf/stoc/MoserS11,DBLP:journals/siamcomp/Hertli14,liu2018ksat}.
As for approximation algorithms,
H{\aa}stad shows that a polynomial-time algorithm with performance ratio greater than $0.875$ would imply \textsf{P} = \textsf{NP} (see Fig.~\ref{poly_inapx} in Appendix~\ref{subsec_poly_table}).
Other witnesses entailing exponential upper bounds include the Linear PCP
Conjecture, which implies the non-existence of a sub-exponential time algorithm for MAX-$3$-SAT with performance ratio greater than $0.875$, assuming ETH \cite{DBLP:conf/iwpec/KimW11}.
Moreover, deterministic algorithms are also essential. Williams shows that a faster deterministic exponential-time algorithm for Circuit-SAT would imply super-polynomial circuit lower bounds for \textsf{NEXP}, which is a notoriously hard open problem \cite{DBLP:journals/siamcomp/Williams13}.
In a word, improving exponential upper bounds for solving NAE-$k$-SAT and MAX-NAE-$k$-SAT approximations
(preferably deterministically) are as crucial as those for $k$-SAT and MAX-$k$-SAT.

Since the 1980's, mounting evidence suggests that NAE-$k$-SAT should be easier than $k$-SAT.
For instance,
the other more-constrained variants called X-SAT (exactly one true literal in each clause) and X-$3$-SAT (and each clause has at most $3$ literals) can be solved in time $1.18^n$ and $1.11^n$ respectively
\cite{DBLP:journals/tcs/ByskovMS05},
which are much better than the fastest $k$-SAT algorithm (see Fig.~\ref{table_result}).
More surprisingly, Planar NAE-$3$-SAT is even in \textsf{P} \cite{moret1988planar}, while Planar $3$-SAT and Planar X-$3$-SAT remain to be \textsf{NP}-complete
\cite{DBLP:journals/siamcomp/Lichtenstein82,DBLP:journals/jal/DyerF86}.
However, before this work, it is unknown whether NAE-$k$-SAT algorithms can be faster than the best $k$-SAT algorithm, at least for deterministic ones.
\footnote{
By the time we submit this work, we are aware of a new result by Hansen et al. that improves the upper bound for Unique NAE-$3$-SAT (the formula guarantees to have exactly two satisfying assignments) to $1.305^n$ \cite{DBLP:conf/stoc/2019}.
The improvement on general NAE-$3$-SAT is even smaller (differs from the bound for $3$-SAT in the fourth decimal digit).
More importantly, their bound is for randomized algorithms.
}

We find out that the NAE property can be used to reduce the searching space of some algorithms, which we vividly call it the \emph{blessing of NAE}, leading to a better upper bound for NAE-$k$-SAT algorithms.
The second blessing of NAE derives from the fact that any clause in an NAE-$3$-SAT instance can be represented by a degree-$2$ polynomial with $3$ variables,
which, counter-intuitively, turns into an algorithm for MAX-NAE-$3$-SAT that runs in time $\text{poly}(n) \cdot 2^{n \omega / 3}$, where $\omega$ is the matrix product exponent.
Note that any length-at-least-$4$ clause does not have this property \cite{williams2007algorithms}.

As a counterpart, there is also a \emph{curse of NAE}.
Considering a clause in a $k$-SAT instance, whenever a literal is assigned to true, we immediately know that this clause must be satisfied.
However, the above is not true in an NAE-$k$-SAT instance since we need at least two literals with different values.
This is annoying because to solve this problem,
it is common to fix the assignments of a subset of variables to eliminate some clauses or to shorten some clauses by at least one, which simplifies the problem by reducing to a smaller formula.
It turns out that such curse ruins a number of existing MAX-$k$-SAT approximation algorithms when applying to MAX-NAE-$k$-SAT, including the current best ones, making MAX-NAE-$k$-SAT much hard than MAX-$k$-SAT at least on $k \ge 4$.

The rest of the paper is organized as follows.
The basic notations are presented in \S{\ref{pre}}.
In \S{\ref{rw_section}} we introduce some related work in algorithms for $k$-SAT and MAX-$k$-SAT approximation, then generalize them to NAE-$k$-SAT and MAX-NAE-$k$-SAT, or point out why the generalizations do not work.
Our result on deterministic algorithms for NAE-$k$-SAT is presented in \S{\ref{det_nae}}, whose comparison results on $3 \le k \le 6$ are highlighted in Fig.~\ref{table_result}, with upper bounds formally stated in Theorem~\ref{main_k_sat_general_form}.
Later in \S{\ref{AA}}, we present our approximation algorithms with upper bounds presented in Theorem~\ref{maxksat} and Theorem~\ref{reduce_solve_analytic}.


\begin{figure*}[t]

\centering
\caption{The rounded up base $c$ in the upper bound $c^n$ of our deterministic algorithm for NAE-$k$-SAT and the corresponding upper bounds in previous results and in the current fastest randomized algorithm PPSZ (marked with $*$).
Since an NAE-$k$-SAT instance is equivalent to a $k$-SAT instance with pairs of opposite-polarity clauses, any $k$-SAT algorithm can solve NAE-$k$-SAT within the same time in terms of the number of variables (omit the polynomial factor).
}\label{table_result}
\centering
\renewcommand{\arraystretch}{1.2}
\begin{tabular}{|c|c|c|c|c|c|}\hline
$k$ & This work & \cite{liu2018ksat} & \cite{DBLP:conf/stoc/MoserS11} & \cite{DBLP:journals/tcs/DantsinGHKKPRS02} & \cite{DBLP:journals/jacm/PaturiPSZ05,DBLP:journals/siamcomp/Hertli14} \\ \hline 
3 & $\mathbf{1.32573}$ & $1.32793$ & $1.33334$ & $1.50001$ & $1.30704^*$ \\ 

4 & $\mathbf{1.49706}$ & $1.49857$ & $1.50001$ & $1.60001$ & $1.46899^*$ \\ 

5 & $\mathbf{1.59888}$ & $1.59946$ & $1.60001$ & $1.66667$ & $1.56943^*$ \\ 

6 & $\mathbf{1.66624}$ & $1.66646$ & $1.66667$ & $1.71429$ & $1.63788^*$ \\ \hline

\end{tabular}
\end{figure*}

\section{Notations}\label{pre}


Let $V=\{v_i \mid i \in [n]\}$ be a set of $n$ Boolean variables.
For all $i \in [n]$, a \emph{literal} $l_i$ is either $v_i$ or $\bar{v}_i$.
A \emph{clause} $C$ is a set of literals and an \emph{instance} $F$ is a set of clauses.
The \emph{occurrence} of a variable $v$ in $F$ is the total number of $v$ and $\bar{v}$ in $F$.
A \emph{$k$-clause} is a clause consisting of exactly $k$ literals, and a \emph{$\le k$-clause} consists of at most $k$ literals.
If every clause in $F$ is a $\le k$-clause, then $F$ is a \emph{$k$-instance}.
An \emph{assignment} $\alpha$ of $F$ is a mapping from $V$ to $\{0,1\}^n$.
A \emph{partial assignment} is the mapping restricted on $V' \subseteq V$ such that only variables in $V'$ are assigned.
A clause $C$ is said to be \emph{satisfied} by $\alpha$ if $\alpha$ assigns at least one literal to $1$ (true) in $C$, and is said to be \emph{NAE-satisfied} if $\alpha$ assigns at least one literal to $1$ and at least one literal to $0$ (false) in $C$.
$F$ is \emph{satisfiable} (resp. \emph{NAE-satisfiable}) if and only if there exists an $\alpha$ satisfying (resp. NAE-satisfying) all clauses in $F$, and we call such $\alpha$ a \emph{satisfying assignment} (resp. \emph{NAE-satisfying assignment}) of $F$.
The $k$-SAT problem asks to find a satisfying assignment of a given $k$-instance, and the NAE-$k$-SAT problems asks to find an NAE-satisfying assignment of it. 
If the context is clear, we will drop the prefix \emph{NAE-}.
The \emph{(NAE-)} in any statement means that the result holds both with and without NAE.

The MAX-SAT problem asks to find an assignment
$\alpha$ of an instance $F$, such that the number of satisfied clauses in $F$ under $\alpha$ is maximized.
If $F$ is a $k$-instance, this problem is called MAX-$k$-SAT.
Analogously we can define the MAX-NAE-$k$-SAT problem.
By definition, a $1$-clause can never be NAE-satisfied, thus when solving the NAE-$k$-SAT or MAX-NAE-$k$-SAT approximation problem on a $k$-instance $F$,
we can safely assume that there is no $1$-clause in $F$.
Given instance $F$, let $s(\alpha)$ be the number of satisfied clauses in $F$ under assignment $\alpha$.
The \emph{optimal assignment} $\alpha^* \coloneqq \arg \max_{\alpha} s(\alpha)$ maximizes the number of satisfied clauses in $F$.
We call $\alpha$ a \emph{$\delta$-approximation assignment}
if $s(\alpha) / s(\alpha^*) \ge \delta$, then $\alpha_{\delta}$ is used to denote such $\alpha$.
Suppose for any $k$-instance on $n$ variables, algorithm $\mathcal{A}$ outputs some $\alpha_{\delta}$ (deterministically or with high probability),
then $\mathcal{A}$ has \emph{performance ratio} $\delta$ and $\mathcal{A}$ is a \emph{$\delta$-approximation algorithm} for MAX-$k$-SAT.
Further if $\mathcal{A}$ runs in $T(n)$ time, we say that MAX-$k$-SAT has a \emph{$T(n)$-time $\delta$-approximation}. Similar definitions work for MAX-NAE-$k$-SAT.

Throughout the paper, $n$ and $m$ always denote the number of variables and number of clauses in the instance respectively and assume $m = \text{poly}(n)$. 
All logarithms are base-two.
We use $\mathcal{O}(T(n)) = 2^{o(n)} \cdot T(n)$ to omit sub-exponential factor in $n$.

\section{Related Work and Generalizations}\label{rw_section}


\subsection{Algorithms for (NAE-)$k$-SAT}\label{rw_det}

Recently, Liu improves the upper bound of deterministic algorithms for $k$-SAT by introducing the concept of \emph{chain} \cite{liu2018ksat}.
Liu's algorithm either solves the formula in desired time or produces a large enough set of chains, which can be used to boost the derandomized local search.
To construct the \emph{generalized covering code} for the derandomized local search, one has to prove its existence, which can be done by solving a specific Linear Programming.
In \S{\ref{det_nae}}, we will give an overview of Liu's method and show a different solution to the Linear Programming for NAE-$k$-SAT, which implies better upper bounds (see Fig.~\ref{table_result}).
Note that the derandomized local search essentially derandomizes Sch\"oning's Random Walk \cite{DBLP:conf/stoc/MoserS11}.
An incremental version of Sch\"oning's Random Walk for MAX-$k$-SAT approximation will be discussed in \S{\ref{rw_apx}}.
We summarize Liu's result in the following theorem:

\begin{theorem}[\cite{liu2018ksat}]\label{liu_ksat}
    There exists a deterministic algorithm for (NAE-)$k$-SAT that runs in time $\mathcal{O}({c_k}^n)$, where
    $ c_3 = 3^{\log{\frac{4}{3}} / \log{\frac{64}{21}}}$, $c_k = (2^k - 1)^{\nu} \cdot {c_{k-1}}^{1 - k \nu}$ for $k \ge 4$ and
    \begin{equation*}
        \nu = \frac{\log(2k - 2) - \log{k} - \log{c_{k-1}}} { \log(2^k - 1) - \log(1 - (\frac{k-2}{2k-2})^k) - k \log{c_{k-1}} }.
    \end{equation*}
\end{theorem}

\subsection{Exact Algorithms for MAX-(NAE-)$k$-SAT}\label{rw_exact}

As shown by Williams in \cite{DBLP:journals/tcs/Williams05}, MAX-$2$-SAT has an exact algorithm that runs in time exponentially less than $\mathcal{O}(2^n)$, which relies on the non-trivial faster algorithm for Matrix Multiplication.
Williams's method also works for MAX-NAE-$2$-SAT and MAX-NAE-$3$-SAT, as noted in the blessing of NAE.
Using the current fastest Matrix Multiplication by Le Gall, the results are presented below:
\begin{theorem}[\cite{DBLP:conf/issac/Gall14a,DBLP:journals/tcs/Williams05,williams2007algorithms}]\label{exact_rw_bound}
    There exist $\mathcal{O}(2^{\omega n / 3})$-time exact algorithms for MAX-$2$-SAT, MAX-NAE-$2$-SAT, and MAX-NAE-$3$-SAT, where $\omega < 2.373$ is the matrix product exponent.
\end{theorem}

This result does not apply to MAX-$3$-SAT or MAX-NAE-$4$-SAT since the longest clause in the instance is equivalent to a degree-at-least-$3$ polynomial,
but
currently Rank-$3$ Tensor Contraction (the generalization of Matrix Multiplication) does not have an $O(n^{4- \epsilon})$-time algorithm for any $\epsilon > 0$.

\subsection{MAX-(NAE-)$k$-SAT Approximations}\label{rw_apx}

The up-to-date polynomial-time MAX-$k$-SAT approximations can be found in Fig.~\ref{poly_inapx} in Appendix~\ref{subsec_poly_table}.
Hirsch gives the first algorithm to approximate MAX-$k$-SAT to arbitrary performance ratio within $(2 - \epsilon)^n$ time \cite{DBLP:journals/dam/Hirsch03}.
We call his algorithm RandomWalk (Algorithm~\ref{MAX_RW}), which is a variant of Sch\"oning's Random Walk for $k$-SAT \cite{DBLP:conf/focs/Schoning99}.
Their differences are summarized in Appendix~\ref{Hirsch_intro}, along with some intuitions behind it.
In \S{\ref{rw_subsection}} we will show that RandomWalk also works for MAX-NAE-$k$-SAT, and the upper bound analyzed by Hirsch can be tighten.
We give the result by Hirsch in the context of MAX-NAE-$k$-SAT:
\begin{theorem}[\cite{DBLP:journals/dam/Hirsch03}]\label{Hirsch_bound}
    MAX-(NAE-)$k$-SAT has an
    $\mathcal{O}((2 - \frac{2 - 2 \delta}{2k - k \delta})^n)$-time $\delta$-approximation.
\end{theorem}

\begin{algorithm}[t]
\caption{RandomWalk}\label{MAX_RW}
\begin{algorithmic}[1]
\REQUIRE $k$-instance $F$
\ENSURE assignment $\hat{\alpha}$
\STATE initialize $\hat{\alpha}$ as an arbitrary assignment in $\{0,1\}^n$
\STATE choose $\alpha$ from $\{0,1\}^n$ randomly \label{rg_line}
\STATE \textbf{repeat} the following for $O(n)$ steps:
    \STATE ~~~~~~\textbf{if} {$s(\alpha) > s(\hat{\alpha})$} \textbf{then} \label{rw_update1}
    \STATE ~~~~~~~~~~~~$\hat{\alpha} \leftarrow \alpha$ \label{rw_update2}
        \STATE ~~~~~~randomly choose an unsatisfied clause $C$ in $F$ \label{clause_choosing}
        \STATE ~~~~~~randomly choose a variable in $C$ and change its value in $\alpha$
\STATE \textbf{return} $\hat{\alpha}$
\end{algorithmic}
\end{algorithm}

Subsequently in \cite{DBLP:journals/tcs/EscoffierPT14}, there are three MAX-$k$-SAT approximation algorithms proposed by Escoffier et al., which are better than Hirsch's algorithm.
Being originally designed for MAX-SAT approximation, their algorithms work for MAX-$k$-SAT approximations as well:
given a $k$-instance as the input, their methods of variables splitting either reduce the problem to exponentially many sub-problems solved by a polynomial-time algorithm one by one, or to a problem with fewer variables without increasing the clause length, or to a problem with weighted and longer clauses.
However, for MAX-NAE-$k$-SAT, as discussed in the curse of NAE, there is no guarantee of the lower bound of satisfied clauses by fixing the assignments of a subset of variables (see details in Appendix~\ref{EPT_intro}).
Therefore all three algorithms based on such reduction by variables splitting do not apply to MAX-NAE-$k$-SAT approximation.
We summarize the results of Escoffier et al. as below, which are currently the best for MAX-$k$-SAT approximations:
\begin{theorem}[\cite{DBLP:journals/tcs/EscoffierPT14}]\label{moderate_apx}
    If there exists an exact algorithm for solving MAX-$k$-SAT that runs in $\mathcal{O}(c^n)$ time,
    then MAX-$k$-SAT has an $\mathcal{O}(\min(2^{n(\delta - \ell) / (1 - \ell)}, c^{n \delta}, 2^{n(2 \delta - 1)}))$-time $\delta$-approximation,
    where $\ell$ is the performance ratio of any given polynomial-time approximation algorithm for MAX-$k$-SAT.
\end{theorem}
Using the current best polynomial-time approximation algorithms (see Fig.~\ref{poly_inapx} in Appendix~\ref{subsec_poly_table}) for the value $\ell$ and the MAX-$2$-SAT exact algorithm for the value $c$ (see Theorem~\ref{exact_rw_bound}), some numerical results are illustrated in Fig.~\ref{fig:max-2-sat} and Fig.~\ref{fig:max-3-sat}.
Their algorithm for MAX-$k$-SAT approximation is faster than ours for MAX-NAE-$k$-SAT except when $\delta$ closes to $1$ and $k=3$.

\section{Deterministic Algorithms for NAE-$k$-SAT}\label{det_nae}

We transform the NAE-$k$-SAT problem to an equivalent $k$-SAT problem: for every clause in the original $k$-instance, create a new clause with opposite polarity in every literal. The new clause is called the \emph{conjugate} of the original clause, and we call these two clauses a \emph{conjugate pair}.
Two conjugate pairs are \emph{independent} if they do not share variables.
Solving the new $k$-instance by a $k$-SAT algorithm would give the same upper bound as for $k$-SAT (omit the $\mathcal{O}(1)$ factor).
However, as we discussed in the blessing of NAE, a conjugate pair has fewer satisfying assignments. We will show how to use this to derive a better upper bound for NAE-$k$-SAT.

We adopt the same algorithmic framework (Algorithm~\ref{Framework}) from \cite{liu2018ksat}, using the conjugate pairs instead of chains.
The subroutine \textsf{BR} (for branching algorithm) is presented later.
One of the key building blocks is to construct the generalized covering code for the subroutine \textsf{DLS} (for derandomized local search, see \cite{liu2018ksat} for details) using Linear Programming.
Recall that for two assignments $a, a^* \in \{0,1\}^k$, the \emph{Hamming distance} $d(a, a^*) = \| a - a^*\|_1$ is the number of disagreed bits.
We have the following:
\begin{algorithm}[t]
\caption{Algorithmic Framework for NAE-$k$-SAT Algorithm}
\label{Framework}
\begin{algorithmic}[1]
\REQUIRE $k$-instance $F$
\ENSURE a satisfying assignment or \verb"Unsatisfiable"
\STATE \textsf{BR}$(F)$ solves $F$ or returns a set of independent conjugate pairs $\mathcal{P}$
\IF {$F$ is not solved}
    \STATE \textsf{DLS}$(F,\mathcal{P})$
\ENDIF
\end{algorithmic}
\end{algorithm}

\begin{lemma}\label{dls_upper_bound}
    Given $k$-instance $F$ and a set of independent conjugate pairs $\mathcal{P}$ from $F$, \textsf{DLS} runs in time $\mathcal{O}((\frac{2(k-1)}{k})^{n - k |\mathcal{P}|} \cdot \lambda^{-|\mathcal{P}|})$,
    where $\lambda$ is the solution to the following Linear Programming $\mathcal{LP}_k$ with variables $\lambda \in \mathbb{R}^+$, $\pi: A \mapsto [0, 1]$, where the \emph{solution space} $A = \{0, 1\}^k \backslash \{0^k, 1^k\}$:
    \begin{align*}
        & \sum_{a \in A} \pi(a) = 1 \\
        & \pi(a) \ge 0 &  \forall a \in A \\
        & \lambda = \sum_{a \in A} \left( \pi(a) \cdot (\frac{1}{k-1}) ^ {d(a, a^*)} \right) & \forall a^* \in A
    \end{align*}
\end{lemma}
The proof of Lemma~\ref{dls_upper_bound} can be found in Appendix~\ref{ksat_lemma}.
In the following lemma, we provide a closed-form solution to $\mathcal{LP}_k$, which leads to the expression of an upper bound for the running time of \textsf{DLS}.
This is different from the Linear Programming proposed in \cite{liu2018ksat} dues to the differences in the solution spaces.
\begin{lemma}\label{lp_solution}
    Given integer $k \ge 3$, the solution to $\mathcal{LP}_k$ is:
    \begin{equation*}
        \lambda = \frac{k^k + (\frac{-k}{k-1})^k}{(2k-2)^k - 2 (k-2)^k + (-2)^k},
    \end{equation*}
    \begin{equation*}
        \pi(a) = \frac{(k-1)^k}{(2k-2)^k - 2 (k-2)^k + (-2)^k} \cdot (1 - (\frac{-1}{k-1})^{d(a, 0^k)}) \cdot (1 - (\frac{-1}{k-1})^{d(a, 1^k)}) \textit{~for all~} a \in A .
    \end{equation*}
\end{lemma}

\begin{proof}
    We prove that this is a feasible solution to $\mathcal{LP}_k$ by verifying all constraints.
    To verify $\lambda \in \mathbb{R}^+$ and $\pi(a) \ge 0~(\forall a \in A)$, it suffices to show that $(2k-2)^k - 2(k-2)^k + (-2)^k > 0$, which is immediate by an induction.

    For verification of $\sum_{a \in A} \pi(a) = 1$, we multiplying $\frac{(2k-2)^k - 2(k-2)^k + (-2)^k}{(k-1)^k}$ on both sides to get:
    \begin{align*}
        \text{LHS} &= \sum_{a \in A}\left( (1 - (\frac{-1}{k-1})^{d(a, 0^k)}) \cdot (1 - (\frac{-1}{k-1})^{d(a, 1^k)}) \right) \\
        &= \sum_{a \in A} \left( 1 - (\frac{-1}{k-1})^{d(a, 0^k)} - (\frac{-1}{k-1})^{d(a, 1^k)} + (\frac{-1}{k-1})^k \right) \\
        &= \sum_{y=1}^{k-1} \binom{k}{y} \cdot \left( 1 - (\frac{-1}{k-1})^y - (\frac{-1}{k-1})^{k - y} + (\frac{-1}{k-1})^k \right) \\
        &= \sum_{y=0}^{k} \binom{k}{y} \cdot \left( 1 - (\frac{-1}{k-1})^y - (\frac{-1}{k-1})^{k - y} + (\frac{-1}{k-1})^k \right) \\
        &= 2^k \cdot \left( 1 + (\frac{-1}{k-1})^k \right) - 2 \cdot (\frac{k-2}{k-1})^k \\
        &= \text{RHS}.
    \end{align*}
    The second equality dues to the relation $d(a, 0^k) + d(a, 1^k) = k$.
    The third equality follows from substituting $d(a, 0^k)$ with $y$ and the fact that the number of $a$'s in $A$ with $d(a, 0^k) = y$ is $\binom{k}{y}$.
    The fourth equality adds two terms with value $0$ to the sum.
    The last equality is the Binomial theorem.

    Similarly, multiplying $\frac{(2k-2)^k - 2(k-2)^k + (-2)^k}{(k-1)^k}$ on both sides of $\lambda = \sum_{a \in A} \left( \pi(a) \cdot (\frac{1}{k-1}) ^ {d(a, a^*)} \right)$ to get:
    \begin{align*}
        \text{RHS} &= \sum_{a \in A}\left( (1 - (\frac{-1}{k-1})^{d(a, 0^k)}) \cdot (1 - (\frac{-1}{k-1})^{d(a, 1^k)}) \cdot (\frac{1}{k-1})^{d(a, a^*)} \right) \\
        &= \sum_{a \in \{0,1\}^k} \left( 1 - (\frac{-1}{k-1})^{d(a, 0^k)} - (\frac{-1}{k-1})^{d(a, 1^k)} + (\frac{-1}{k-1})^k \right) \cdot (\frac{1}{k-1})^{d(a, a^*)} \\
        &= (\frac{k}{k-1})^k + (\frac{-1}{k-1})^k \cdot (\frac{k}{k-1})^k \\
        &- \sum_{a \in \{0,1\}^k} \left( (-1)^{d(a,0^k)} \cdot (\frac{1}{k-1})^{d(a,0^k) + d(a, a^*)} \right) \\
        &- \sum_{a \in \{0,1\}^k} \left( (-1)^{d(a,1^k)} \cdot (\frac{1}{k-1})^{d(a,1^k) + d(a, a^*)} \right)\\
        &= (\frac{k}{k-1})^k + (\frac{-1}{k-1})^k \cdot (\frac{k}{k-1})^k \\
        &= \text{LHS}.
    \end{align*}
    The third equality is by noticing that $a^*$ is symmetric with respect to $0^k$ and $1^k$ and applying the Binomial theorem.
    For the fourth equality to hold we need to prove that the sums in the fourth line and fifth line are equal to $0$.
    By symmetry we only need to prove that
    \begin{equation}
        \sum_{a \in \{0,1\}^k} \left( (-1)^{d(a,0^k)} \cdot (\frac{1}{k-1})^{d(a,0^k) + d(a, a^*)} \right) = 0. \label{term_value_0}
    \end{equation}
    Observe that for $a^*$ to be an NAE-assignment,
    there must exist $\tilde{i} \in [k]$ such that the $\tilde{i}$-th bit $a^*_{\tilde{i}} = 1$.
    Now we partition $\{0,1\}^k$ into two sets $S_0, S_1$ depending on whether the $\tilde{i}$-th bit is $1$.
    We observe the following bijection: for each $a \in S_0$, negate the $\tilde{i}$-th bit to get $a' \in S_1$.
    Then it must be that
    \begin{equation*}
        d(a, 0^k) + d(a, a^*) = (d(a', 0^k) + 1) + (d(a', a^*) - 1)
    \end{equation*}
    since $a^*_{\tilde{i}} = a_{\tilde{i}} = 1 - a'_{\tilde{i}} = 1$.
    As a result, for each $a \in S_0$, there is $a' \in S_1$ such that the corresponding terms of $a, a'$ in sum (\ref{term_value_0}) have the same exponent on $\frac{1}{k-1}$
    and $(-1)^{d(a, 0^k)} = - (-1)^{d(a', 0^k)}$,
    so the sum of such pair is $0$ and (\ref{term_value_0}) must be $0$.
    Therefore we verified all the constraints and proved the lemma.
 \end{proof}


We present \textsf{BR} in Algorithm~\ref{br_k}, with parameter $\nu$ to be fixed in Theorem~\ref{main_k_sat_general_form}.
What is different from the branching algorithm in \cite{liu2018ksat} is line~\ref{line_enumerate}, which is the blessing of NAE: without NAE, we cannot exclude $1^k$ from the satisfying assignments.
After fixing all variables in $\mathcal{P}$, the remaining formula is a $(k-1)$-instance due to the maximality of $\mathcal{P}$.
Note that in line~\ref{line_call_k_1} we cannot call a deterministic NAE-$(k-1)$-SAT algorithm because the remaining formula is not necessarily consisting of only conjugate pairs, i.e., it might no longer be equivalent to an NAE instance dues to the curse of NAE.

\begin{algorithm}[t]
\caption{Algorithm \textsf{BR}}
\label{br_k}
\begin{algorithmic}[1]
\REQUIRE $k$-instance $F$, parameter $\nu$
\ENSURE a satisfying assignment or \verb"Unsatisfiable" or a set of independent conjugate pairs $\mathcal{P}$
\STATE greedily construct a maximal set of independent conjugate pairs $\mathcal{P}$
\IF {$|\mathcal{P}| < \nu n$}
    \FOR {each assignment $\alpha \in \{\{0, 1\}^k \backslash \{0^k, 1^k\}\}^{|\mathcal{P}|}$ of $\mathcal{P}$} \label{line_enumerate}
        \STATE solve the remaining formula by a deterministic $(k-1)$-SAT algorithm after fixing $\alpha$ in $F$ \label{line_call_k_1}
        \STATE \textbf{return} the satisfying assignment if satisfiable
    \ENDFOR
    \RETURN \verb"Unsatisfiable"
\ELSE
    \RETURN $\mathcal{P}$
\ENDIF
\end{algorithmic}
\end{algorithm}

Obviously, \textsf{BR} runs in time $\mathcal{O}((2^k - 2)^{|\mathcal{P}|} \cdot c_{k-1}^{n - k |\mathcal{P}|})$ where $c_{k-1}$ is the base of the exponential upper bound of a given deterministic $(k-1)$-SAT algorithm.
Substitute the values of $c_k$ from Theorem~\ref{liu_ksat} and $\lambda$ from Lemma~\ref{lp_solution} into Lemma~\ref{dls_upper_bound}, we obtain that the running time of \textsf{BR} is an increasing function of $|\mathcal{P}|$, while the running time of \textsf{DLS} is a decreasing function of $|\mathcal{P}|$.
This immediately implies that the worst case is attained when two running times are equal, which gives the following main result:

\begin{theorem}[Result on Deterministic Algorithm]\label{main_k_sat_general_form}
    Given integer $k \ge 3$,
    if there exists a deterministic algorithm for $(k - 1)$-SAT that runs in time $\mathcal{O}({c_{k-1}}^n)$,
    then there exists a deterministic algorithm for NAE-$k$-SAT that runs in time $\mathcal{O}({c'_k}^n)$, where
    \begin{equation*}
        c'_k = (2^k - 2)^{\nu} \cdot {c_{k-1}}^{1 - k \nu}
    \end{equation*}
    and
    \begin{equation*}
        \nu = \frac{\log(2k - 2) - \log{k} - \log{c_{k-1}}} { \log(2^k - 2) + \log(k^k + (\frac{-k}{k-1})^k) + k \log(\frac{2k-2}{k c_{k-1}}) - \log((2k-2)^k - 2 (k-2)^k + (-2)^k)}.
    \end{equation*}
\end{theorem}
Using the values of $c_k$ from Theorem~\ref{liu_ksat}, we obtain all the upper bound results in Fig.~\ref{table_result}.

\section{Approximation Algorithms}\label{AA}

In this section, we present two approximation algorithms: RandomWalk and ReduceSolve, both of which work for MAX-$k$-SAT and MAX-NAE-$k$-SAT approximations.
RandomWalk can be repeated for exponential times to obtain an approximation assignment with high probability.
\footnote{Suppose the $O(n)$-step RandomWalk succeeds with probability $p$, then one can repeat it for $\mathcal{O}(1/p)$ times to get an $\mathcal{O}(1/p)$-time randomized algorithm with high probability of success.
If no ambiguity, we still call such repeating algorithm RandomWalk.}
ReduceSolve transforms the instance to another instance with fewer variables and then solves the remaining formula exactly.

\subsection{Algorithm RandomWalk}\label{rw_subsection}

For those readers familiar with Sch\"oning's Random Walk \cite{DBLP:conf/focs/Schoning99},
it is obvious that Algorithm~\ref{MAX_RW} works for MAX-NAE-$k$-SAT approximation: just be aware that in line~\ref{clause_choosing}
an NAE-unsatisfied clause is chosen randomly.
We use part of Hirsch's result on Algorithm~\ref{MAX_RW} and then give a tighter analysis.



\begin{lemma}[\cite{DBLP:journals/dam/Hirsch03}]\label{dec_lemma}
    For any $k$-instance $F$, RandomWalk returns a $\delta$-approximation assignment (resp. NAE-assignment) of $F$ with probability at least $(2 - 2 p_{\delta})^{-n}$,
    where $p_{\delta} = \frac{1 - \delta}{k (m / s(\alpha^*) - \delta)}$ and $\alpha^*$ is an optimal assignment (resp. NAE-assignment) of $F$.
\end{lemma}

To give a tighter bound, we need the following lemma to bound $m$ and $s(\alpha^*)$
by introducing the \emph{average clause length} $\eta \coloneqq (\sum_{i \in [k]} i \cdot m_i) / m$,
where $m_i$ is the number of $i$-clauses ($i \in [k]$) in $F$.

\begin{lemma}\label{km}
    Given $k$-instance $F$, let $\eta$ be its average clause length.
    If $\alpha^*$ is an optimal assignment of $F$, it must be that $m \le \frac{s(\alpha^*)}{\xi}$, where $\xi = \frac{2^{k-1} (\eta+k-2) - \eta + 1}{2^k (k-1)}$.
    If $\alpha^*$ is an optimal NAE-assignment of $F$, it must be that $m \le \frac{s(\alpha^*)}{\xi'}$, where $\xi' = \frac{1}{2}$ for $k=2$ and $\xi' = \frac{2^{k-1} (\eta+k-4) - 2\eta + 4}{2^k (k-2)}$ for $k \ge 3$.
\end{lemma}

\begin{proof}
    It is easy to see that a random assignment $\alpha$ satisfies $\sum_{i \in [k]} \frac{2^i - 1}{2^i} m_i$ clauses in expectation,
    so it must be that $s(\alpha^*) \ge \sum_{i \in [k]} \frac{2^i - 1}{2^i} m_i$.
    Since $\eta = (\sum_{i \in [k]} i \cdot m_i) / m$,
    we can eliminate $m_1$ by $\eta$. By substitution, multiplication and rearranging we have:
    \begin{equation*}
        \frac{s(\alpha^*)}{m} \ge \frac{\sum_{i=2}^k (2^{i-1} (\eta+i-2) - \eta + 1) m_i}{\sum_{i=2}^k 2^i (i - 1) m_i} \ge \frac{2^{k-1} (\eta+k-2) - \eta + 1}{2^k (k-1)}.
    \end{equation*}
    For the last inequality to hold
    we need to prove that
    \begin{equation*}
        \min_{2 \le i \le k} \{ \frac{2^{i-1} (\eta+i-2) - \eta + 1}{2^i (i-1)} \} = \frac{2^{k-1} (\eta+k-2) - \eta + 1}{2^k (k-1)}
    \end{equation*}
    holds for $\eta \ge 1$.
    Let $f_{\eta}(i) = \frac{2^{i-1} (\eta+i-2) - \eta + 1}{2^i (i-1)}$, we prove that
    \begin{equation}
        f_{\eta}(i) \ge f_{\eta}(i + 1) \label{ffi}
    \end{equation}
    holds for $2 \le i \le k-1$.
    Simple computations give us that (\ref{ffi}) holds when $2^{2i} \ge 2^i (i + 1)$ or $\eta = 1$, which holds for $i \ge 2$.
    Thus we proved that the minimal of $f_{\eta}(i)$ is attained when $i$ is maximized, i.e., $i = k$.

    The statement for NAE-assignment can be proved in a similar way:
    since there is no $1$-clause, we know that $s(\alpha^*) \ge \sum_{i=2}^k \frac{2^i - 2}{2^i} m_i$; eliminating $m_2$ by $\eta$ for $k \ge 3$, the lemma follows from a similar computation.
 \end{proof}

Based on Lemma~\ref{km},
one can easily show that $s(\alpha^*) \ge m / 2$ for MAX-(NAE-)$k$-SAT, thus Lemma~\ref{dec_lemma} immediately implies Theorem~\ref{Hirsch_bound}.


Our key observation is that although the probability of  success of RandomWalk is an increasing function of $\eta$,
a random guess (line~\ref{rg_line} of Algorithm~\ref{MAX_RW}) is actually not too bad when $\eta$ is small: there are many short clauses in the formula, then the optimal solution cannot satisfy too many clauses, and a random guess should not be too far away from it.
Indeed, a random guess
already yields an arbitrarily good approximation with non-negligible probability.
To show this, we first take a detour to focusing on a special subformula:

\begin{definition}\label{mss}
    Given $k$-instance $F$ and an optimal (NAE-)assignment $\alpha^*$, define its \emph{maximal (NAE-)satisfiable subformula} as $G$ being a $k$-instance consisting of all (NAE-)satisfied clauses of $F$ under $\alpha^*$.
\end{definition}
Clearly $G$ has $n$ variables, because otherwise assigning a variable outside $G$ can (NAE-)satisfy more clauses of $F$.
Analogous to what we defined for $F$,
let $w_i$ be the number of $i$-clauses in $G$ for all $i \in [k]$ and let $w \coloneqq \sum_{i \in [k]} w_i = s(\alpha^*)$, then the average clause length of $G$ is $\theta \coloneqq (\sum_{i \in [k]} i \cdot w_i) / w$.

\begin{lemma}\label{ab_lemma}
    Given $G, \theta, w$ as defined above,
    let $B_{\tau}(G)$ be the set of all variables whose occurrences are upper bounded by $\tau$.
    For any $\lambda > 0$, it must be that $|B_{\tau}(G)| \ge (1 - \frac{\theta}{\lambda}) n$, where $\tau = \frac{\lambda w}{n}$.
\end{lemma}
The proof is by Markov's inequality and omitted.
With Lemma~\ref{ab_lemma} we obtain
the success probability of random guess (part of RandomWalk, line~\ref{rg_line} of Algorithm~\ref{MAX_RW}) in Lemma~\ref{random_guess} as below: 

\begin{lemma}\label{random_guess}
    For any $k$-instance, RandomWalk returns a $\delta$-approximation (NAE-)assignment with probability at least $2^{- \frac{\theta n}{1 - \delta + \theta}}$.
\end{lemma}

\begin{proof}
    The probability of success of RandomWalk is lower bounded by that of random guess (line~\ref{rg_line} of Algorithm~\ref{MAX_RW}).
    By Lemma~\ref{ab_lemma} with $\lambda = 1 - \delta + \theta$,
    within $G$
    there exist at least $(1 - \frac{\theta} {1 - \delta + \theta}) n$ variables whose occurrences are upper bounded by $\frac{(1 - \delta+ \theta) w}{n}$.
    We arbitrarily choose $(1 - \frac{\theta}{1 - \delta + \theta}) n$ of them and call these variables \emph{sub-$\tau$ variables}, since they constitute a subset of $B_{\tau}(G)$.
    Note that the total occurrences of sub-$\tau$ variables in $G$ is at most:
    \begin{equation*}
        (1 - \frac{\theta}{1 - \delta + \theta}) n \cdot \frac{(1 - \delta + \theta) w}{n} = (1 - \delta) w.
    \end{equation*}
    This is the maximal number of clauses in which sub-$\tau$ variables can occur.
    We have that at least $\delta w$ clauses do not contain any sub-$\tau$ variable.
    So no matter how to change the assignments of sub-$\tau$ variables in $\alpha^*$ and let $\alpha_{\delta}$ be the altered $\alpha^*$, $\alpha_{\delta}$ still satisfies at least $\delta w$ clauses.

    Now randomly guessing an $\alpha$ from $\{0, 1\}^n$, it agrees with some $\alpha_{\delta}$ with probability at least $2^{- \frac{\theta n}{1 - \delta + \theta}}$, because the number of variables which are not sub-$\tau$ variable is $\frac{\theta n}{1 - \delta + \theta}$.
    The conclusion follows immediately.
    It is easy to verify that our choice of $\lambda$ is optimal for maximizing the lower bound of this probability.
 \end{proof}


A tighter upper bound for MAX-(NAE-)$k$-SAT relies on the following correlation of $\eta$ and $\theta$:

\begin{lemma}\label{eta_theta_relation}
    Given $k$-instance $F$ ($k \ge 2$) and $\eta, \theta > 1$ as defined above,
    it must be that $\frac{2}{\theta - 1} \ge \frac{1}{\eta - 1} + \frac{1}{k - 1} \frac{2^{k - 1} - 1}{2^{k - 1}}$.
    Let $\theta'$ be the average clause length of the maximal NAE-satisfiable subformula of $F$, it must be that
    $\theta' = 2$ when $k = 2$ and
    $\frac{2}{\theta' - 2} \ge \frac{1}{\eta - 2} + \frac{1}{k - 2} \frac{2^{k - 2} - 1}{2^{k - 2}}$ when $k \ge 3$.
\end{lemma}

\begin{proof}
    Using $\eta = \frac{\sum_{i \in [k]} i \cdot m_i}{m}$ and $\theta = \frac{\sum_{i \in [k]} i \cdot w_i}{w}$ to eliminate $m_1$ and $w_1$, by $w \ge \sum_{i \in [k]} \frac{2^i - 1}{2^i} m_i$ (see the proof of Lemma~\ref{km}),
    and using the fact that $\forall i \in [k], w_i \le m_i$ for the left-hand side, and rearranging the right-hand side,
    we have:
    \begin{equation}
        \sum_{i = 2}^k \frac{i - 1}{\theta - 1} m_i \ge \sum_{i = 2}^k (\frac{i - 1}{2 \eta - 2} + \frac{1}{2} - \frac{1}{2^i}) m_i \label{eta_theta}.
    \end{equation}
    Now we prove that $\frac{k - 1}{\theta - 1} \ge \frac{k - 1}{2 \eta - 2} + \frac{1}{2} - \frac{1}{2^k}$.
    Assume for contradiction that $\frac{k - 1}{\theta - 1} < \frac{k - 1}{2 \eta - 2} + \frac{1}{2} - \frac{1}{2^k}$ ($k \ge 3$), it must be that:
    \begin{equation*}
        \frac{k - 2}{\theta - 1}  < (\frac{k - 1}{2 \eta - 2} + \frac{1}{2} - \frac{1}{2^k}) \cdot \frac{k - 2}{k - 1}
        = \frac{k - 2}{2 \eta - 2} + \frac{1}{2} - \frac{1}{2^{k - 1}} \cdot \frac{2^{k - 2} + \frac{k - 2}{2}}{k - 1}
        \le \frac{k - 2}{2 \eta - 2} + \frac{1}{2} - \frac{1}{2^{k - 1}}.
    \end{equation*}
    We can continue this process to get $\frac{i - 1}{\theta - 1} < \frac{i - 1}{2 \eta - 2} + \frac{1}{2} - \frac{1}{2^i}$ for all $2 \le i \le k$. This is a contradiction since all coefficients of $m_i$ in the left-hand side of (\ref{eta_theta}) are strictly smaller than those in the right-hand side, so inequality (\ref{eta_theta}) does not hold unless $m_i = 0$ for all $2 \le i \le k$, which is a contradiction since $\eta > 1$.
    As a result, we have $\frac{k - 1}{\theta - 1} \ge \frac{k - 1}{2 \eta - 2} + \frac{1}{2} - \frac{1}{2^k}$, which is
    $\frac{2}{\theta - 1} \ge \frac{1}{\eta - 1} + \frac{1}{k - 1} \frac{2^{k - 1} - 1}{2^{k - 1}}$ as stated.

    Similar argument gives us the result for NAE in this lemma,
    and the case $k=2$ follows directly from the fact that there is no $1$-clause.
 \end{proof}


Since $\theta$ is upper bounded by some function of $\eta$ ,
we have that the worst case is attained when the equalities in Lemma~\ref{eta_theta_relation} are attained.
Now the probability in Lemma~\ref{dec_lemma} is an increasing function of $\eta$, while the probability in Lemma~\ref{random_guess} is a decreasing function of $\eta$,
thus the worst case is attained when they are equal.
So we have our main result on MAX-NAE-$k$-SAT as the following.
The result on MAX-$k$-SAT can be obtained in the same way, which is omitted due to limitation of space.

\begin{theorem}[Result on RandomWalk]\label{maxksat}

    MAX-NAE-$k$-SAT has an $\mathcal{O}(\gamma'^n)$-time $\delta$-approximation, where $\gamma'$ satisfies the following equation system $\mathcal{M'}_k(\delta)$:
    \begin{equation*}
    \left \{
    \begin{aligned}
        & \xi' = \frac{2^{k-1} (\eta+k-4) - 2\eta + 4}{2^k (k-2)} \text{~for~}k \ge 3\ \text{~and~}\xi'=\frac{1}{2}\text{~for~}k=2 \\
        & \gamma' =  2^{\frac{\theta'}{1 - \delta + \theta'}} = 2 - \frac{2 \xi' - 2 \xi' \delta}{k - \xi' \delta k} \\
        & \frac{2}{\theta' - 2} = \frac{1}{\eta - 2} + \frac{1}{k - 2} \frac{2^{k - 2} - 1}{2^{k - 2}}
    \end{aligned}
    \right.
    \end{equation*}
    where integer $k \ge 2$ and constant $\delta \in [0,1]$ are given, and $\gamma', \xi', \theta', \eta$ are variables of $\mathcal{M'}_k(\delta)$.
\end{theorem}
By monotonicity with respect to $\eta$,
a binary search solves $\mathcal{M'}_k(\delta)$ to any given precision in reasonable time.
Some numerical results as our upper bounds for $k \ge 4$ and performance ratio $\delta$ are illustrated in Fig.~\ref{fig:max-4-sat} and Fig.~\ref{fig:max-5-sat} in Appendix~\ref{sec_comp}.

\subsection{Algorithm ReduceSolve}\label{sat_reduction}

Our second approximation algorithm ReduceSolve (Algorithm~\ref{max_sk_sat_alg2}) is to reduce the formula to another formula with fewer variables and solve it by an exact algorithm for MAX-(NAE-)$k$-SAT.
The high-level idea is the following: deliberately choose variables with low occurrence, such that these variables can be fixed without falsifying too many clauses, then solving the reduced formula still yields a good approximation.

\begin{algorithm}[h]
\caption{ReduceSolve}\label{max_sk_sat_alg2}
\begin{algorithmic}[1]
\REQUIRE $k$-instance $F$, parameter $t$
\ENSURE assignment $\hat{\alpha}$
\STATE initialize $V_e \leftarrow \emptyset$, $F_e = \emptyset$ \label{line_start}
\STATE \textbf{for} $i \leftarrow 1$ to $t$ \textbf{do} \label{loop_begin}
\STATE ~~~~~~choose the variable $v$ in $F$ with the lowest occurrence
\STATE ~~~~~~$V_e \leftarrow V_e \cup \{v\}$
\STATE ~~~~~~\textbf{for} every clause $C$ of $F$ containing $v$ \textbf{do}
\STATE ~~~~~~~~~~~~$F_e \leftarrow F_e \cup C$
\STATE ~~~~~~~~~~~~eliminate $C$ from $F$ \label{loop_end}
\STATE solve $F$ by an exact algorithm for MAX-(NAE-)$k$-SAT to get $\alpha^{(1)}$ \label{line_exact}
\STATE $\alpha^{(2)} \leftarrow$ a random partial assignment on $V_e$ \label{half_sat}
\STATE \textbf{return} $\hat{\alpha} \leftarrow \alpha^{(1)} \cup \alpha^{(2)}$
\end{algorithmic}
\end{algorithm}

Line \ref{half_sat} (NAE-)satisfies at least half of the clauses in $F_e$ in expectation,
because every clause in $F_e$ contains at least one variable from $V_e$.
Using the method of conditional probabilities (see Chapter 16 in \cite{alon2016probabilistic}), it is guaranteed to find an $\alpha^{(2)}$ in polynomial time (NAE-)satisfying at least half of the clauses, which makes our algorithm deterministic, provided that the exact algorithm is deterministic (which is the case for Williams's algorithm).

We present our main result on ReduceSolve for MAX-NAE-$k$-SAT (Theorem~\ref{reduce_solve_analytic}).
The result on MAX-$k$-SAT can be obtained in the same way, which is omitted due to limitation of space.

\begin{theorem}[Result on ReduceSolve]\label{reduce_solve_analytic}

    If there exists an exact algorithm for solving MAX-NAE-$k$-SAT that runs in $\mathcal{O}(c^n)$ time,
    then MAX-NAE-$k$-SAT has an $\mathcal{O}(c^{n (1 - 2(1-\delta)\xi') ^{\frac{1}{k}}})$-time $\delta$-approximation where $\xi' = \frac{1}{2}$ for $k=2$ and $\xi' = \frac{2^{k-1} (\eta+k-4) - 2\eta + 4}{2^k (k-2)}$ for $k \ge 3$.
\end{theorem}

\begin{proof}
    We analyze Algorithm~\ref{max_sk_sat_alg2}. In the following, step $i$ corresponds to the loop variable $i$ in line \ref{loop_begin}-\ref{loop_end}.
    Let $F^{(i)}$ be the remaining formula after elimination in step $i$, and let $m^{(i)}$ be the number of clauses in $F^{(i)}$.
    Clearly there are $n - i$ variables in $F^{(i)}$. Taking the average we have that the lowest occurrence is upper bounded by $\frac{k m^{(i)}}{n - i - 1}$.
    (It is possible that there are fewer than $n - i$ variables in step $i$ when a variable's all occurrences are eliminated by other variables, then we still think that there are $n - i$ variables with some variable's occurrence being $0$, and the analysis still holds.)
    The number of clauses in $F^{(i + 1)}$ is $m^{(i+1)} \ge m^{(i)} - \frac{k m^{(i)}}{n - i -1}$, because the lowest-occurrence variable occurs in at most $\frac{k m^{(i)}}{n - i -1}$ clauses. Expanding until $m^{(0)} = m$, the following must hold for the number of clauses in $F^{(t)}$:
    \begin{equation*}
        m^{(t)} \ge m \cdot \prod_{i = 1}^{t} (1 - \frac{k}{n - i}) .
    \end{equation*}
    Using the fact that $1 - y = \exp(-y - o(y))$ for $y \rightarrow 0$, we have:
    \begin{equation}
        m^{(t)} \ge m \cdot \exp(-\sum_{i = 1}^{t} \frac{k}{n - i} - o(\sum_{i = 1}^{t} \frac{k}{n - i})) \label{rec2} .
    \end{equation}
    Note that $\sum_{i = 1}^{t} \frac{1}{n - i} = \ln{\frac{n}{n - t}} + O(\frac{1}{n})$. Now assuming $t = \Theta(n)$, we have that $\ln{\frac{n}{n - t}} = \Theta(1)$, so (\ref{rec2}) becomes:
    \begin{equation}
        m^{(t)} \ge m \cdot \exp(-k \ln{\frac{n}{n - t}} - o(1)) = m \cdot (\frac{n}{n - t})^{-k} \cdot (1 - o(1)) \label{rec3} .
    \end{equation}
Let $x < 1$ be a parameter to be fixed later.
Observe that if $m^{(t)} \ge (2 x - 1) m$, we eliminated at most $(2 - 2 x) m$ clauses in the first $t$ steps, and at most $(1 - x) m$ of them are unsatisfied.
As a result, to obtain an assignment falsifying at most $(1 - x) m$ clauses in $F_e$, by (\ref{rec3}) it is sufficient to have:
\begin{equation*}
    m \cdot (\frac{n}{n - t})^{-k} \cdot (1 - o(1)) \ge (2 x - 1) m ,
\end{equation*}
which can be implied by
\begin{equation*}
    t \le (1 - (2 x - 1) ^ {1 / k} - o(1)) n.
\end{equation*}
Choosing $ t = (1 - (2 x - 1) ^ {1 / k} - o(1)) n$, which is $t = \Theta(n)$ as we assumed for (\ref{rec3}), we have that the variables in the remaining formula $F^{(t)}$ is at most $n - t = (2 x - 1) ^ {1 / k} n + o(n)$.

Now we fix parameter $x$.
If at most $(1 - \delta) s(\alpha^*)$ clauses in the maximal NAE-satisfiable subformula are falsified, we definitely have a $\delta$-approximation assignment. In the worst case, all $(1 - x) m$ clauses we falsified are from the maximal NAE-satisfiable subformula, which gives
\begin{equation*}
    (1 - x) m \le (1 - \delta) s(\alpha^*)
\end{equation*}
to meet the condition.
By Lemma~\ref{km} it suffices to have $x = 1 - (1 - \delta) \xi'$, where
\begin{equation*}
    \xi' = \frac{2^{k-1} (\eta+k-4) - 2\eta + 4}{2^k (k-2)} ~\text{for}~ k \ge 3 ~\text{and}~ \xi' = \frac{1}{2} ~\text{for}~ k=2 .
\end{equation*}
Finally, we solve the remaining formula $F^{(t)}$ by an exact algorithm (line \ref{line_exact}) in time
\begin{equation*}
    \mathcal{O}(c^{n (1 - 2(1 - \delta) \xi')^{1 / k} + o(n)}) = \mathcal{O}(c^{n (1 - 2(1 - \delta) \xi')^{1 / k}}) .
\end{equation*}

If we find an optimal assignment $\alpha^{(1)}$ for $F^{(t)}$, by union with $\alpha^{(2)}$ we obtain a $\delta$-approximation assignment, because $\alpha^{(1)}, \alpha^{(2)}$ are on disjoint variables.
Obviously, line \ref{line_start}-\ref{loop_end} and line~\ref{half_sat} run in polynomial time, giving the theorem.
 \end{proof}

Using Theorem~\ref{exact_rw_bound} for the value $c$,
some numerical results on $k \le 3$ are presented in Fig.~\ref{fig:max-2-sat} and Fig.~\ref{fig:max-3-sat} in Appendix~\ref{sec_comp}.

\newpage

\bibliographystyle{alpha}
\bibliography{nae}

\newpage

\appendix

\section{Some Related Approximation Algorithms}

\subsection{Polynomial-time Approximations}\label{subsec_poly_table}


\begin{figure*}[h]
\centering
\renewcommand{\arraystretch}{1.2}
\caption{The lower bound denotes the current best performance ratios of polynomial-time approximation algorithms; the upper bound denotes the inapproximable thresholds unless \textsf{P} = \textsf{NP}. MAX-E-NAE-$3$-SAT restricts each clause to have exactly three literals.
Assuming the Unique Games Conjecture true, many upper bound results in this figure can be improved \cite{DBLP:conf/stoc/Austrin07,DBLP:conf/stoc/Raghavendra08}.
}
\label{poly_inapx}
\centering
\begin{tabular}{|c|c|c|}\hline
& Lower bound & Upper bound \\\hline
 MAX-$2$-SAT & $0.940$ \cite{DBLP:conf/ipco/LewinLZ02} & $0.955$ \cite{DBLP:journals/jacm/Hastad01} \\
 MAX-$3$-SAT & $0.875$ \cite{DBLP:conf/soda/Zwick02a} & $0.875$  \cite{DBLP:journals/jacm/Hastad01} \\
 MAX-$4$-SAT & $0.872$ \cite{DBLP:journals/jal/HalperinZ01a} & $0.875$ \cite{DBLP:journals/jacm/Hastad01}\\
 MAX-$k$-SAT & $0.8434$ \cite{DBLP:conf/waoa/AvidorBZ05} & $0.875$  \cite{DBLP:journals/jacm/Hastad01} \\\hline
 MAX-NAE-$2$-SAT & $0.878$ \cite{DBLP:journals/jacm/GoemansW95} & $0.917$ \cite{DBLP:journals/siamcomp/TrevisanSSW00} \\
 MAX-E-NAE-$3$-SAT & $0.878$ \cite{DBLP:conf/soda/Zwick98} & $0.978$  \cite{DBLP:conf/soda/Zwick98}\\
 MAX-NAE-$k$-SAT & $0.8279$ \cite{DBLP:conf/waoa/AvidorBZ05} & $0.875$  \cite{DBLP:journals/jacm/Hastad01}   \\\hline
\end{tabular}
\end{figure*}

\subsection{Algorithm by Hirsch}\label{Hirsch_intro}

RandomWalk (Algorithm~\ref{MAX_RW}) modifies Sch\"oning's Random Walk for $k$-SAT \cite{DBLP:conf/focs/Schoning99} in the following ways:
(\romannumeral1) $\hat{\alpha}$ is iteratively updated as the assignment satisfying the most number of clauses so far (lines \ref{rw_update1}-\ref{rw_update2} of Algorithm~\ref{MAX_RW}); (\romannumeral2) choosing a random unsatisfied clause instead of an arbitrary one (line \ref{clause_choosing} of Algorithm~\ref{MAX_RW}).
The motivation is that we do not know the optimal assignment in advance, so there is no termination condition as for $k$-SAT (hitting an assignment satisfying all clauses).
But note that what we ask for is only an approximation, thus it remains hopeful to work out as for $k$-SAT by choosing a random unsatisfied clause to decrease the Hamming distance to a target approximation assignment.

\subsection{Algorithms by EPT}\label{EPT_intro}

Now we introduce the algorithms proposed by Escoffier et al. in \cite{DBLP:journals/tcs/EscoffierPT14} for MAX-SAT approximations and generalize them to MAX-$k$-SAT approximations.

\paragraph{Algorithm $\mathcal{A}$.}
Let $p,q$ be two integers such that $p / q = (\delta - \ell) / (1 - \ell)$, where $\ell$ is the performance ratio of any given polynomial-time approximation algorithm (e.g., one of the algorithms which give the lower bounds in Fig.~\ref{poly_inapx} in Appendix~\ref{subsec_poly_table}).
Build $q$ subsets of variables, each one includes $np/q$ variables, where each variable occurs in exactly $p$ subsets.
For each subset, enumerating all possible truth assignments on its variables and run the polynomial-time approximation algorithm on the remaining formula.
Return the complete assignment with the maximum number of satisfying clauses.
For MAX-$k$-SAT approximations, it can be shown that $\mathcal{A}$ has performance ratio $\delta$ and runs in time $\mathcal{O}(2^{n(\delta - \ell) / (1 - \ell)})$.
This relies on the fact that every clause in the optimal solution contains at least one true literal.
However, for MAX-NAE-$k$-SAT, as discussed in the curse of NAE, there is no guarantee of the lower bound of satisfied clauses by fixing the assignments of a subset of variables.
Therefore $\mathcal{A}$ does not apply to MAX-NAE-$k$-SAT approximation.

\paragraph{Algorithm $\mathcal{B}$.}
Let $p,q$ be two integers such that $p / q = \delta$.
Build $q$ subsets of variables as above.
Remove from the instance the variables not in this subset and all empty clauses.
Solve the remaining instance by an exact algorithm for MAX-SAT and complete this assignment with arbitrary truth values.
It is not hard to see
that $\mathcal{B}$ for MAX-$k$-SAT approximation has performance ratio $\delta$ and runs in time $\mathcal{O}(c^{n \delta})$, where $\mathcal{O}(c^n)$ is the upper bound of the exact MAX-$k$-SAT algorithm.
As discussed in the curse of NAE, $\mathcal{B}$ does not work for MAX-NAE-$k$-SAT approximation for the same reason above.

\paragraph{Algorithm $\mathcal{C}$.}
Let $p,q$ be two integers such that $p / q = 2 \delta - 1$.
Build $q$ subsets of variables as above.
For each subset, assign weight $2$ to every clause containing only variables in this subset, and weight $1$ to every clause containing at least one other variable not in this subset.
Remove from the instance the variables not in this subset and all empty clauses.
Solve the remaining instance using an exact algorithm for weighted MAX-SAT.
Escoffier et al. show that $\mathcal{C}$ for MAX-$k$-SAT approximation has performance ratio $\delta$ and runs in time $\mathcal{O}(c^{n (2\delta - 1)})$, where $\mathcal{O}(c^n)$ is the upper bound of the exact algorithm for solving weighted MAX-$(k+1)$-SAT.
But for the same reason behind previous algorithms, this does not work for MAX-NAE-$k$-SAT approximation.
Currently we do not know any exact algorithm for (weighted) MAX-$k$-SAT with better than $\mathcal{O}(2^n)$ running time when $k \ge 3$, therefore $\mathcal{O}(2^{n (2 \delta - 1)})$ is the upper bound for $\mathcal{C}$.


\section{Proof of Lemma~\ref{dls_upper_bound}}\label{ksat_lemma}

First of all, we review the definition of \emph{chain} and the lemma for the running time of algorithm \textsf{DLS} in \cite{liu2018ksat}.
\begin{definition}[\cite{liu2018ksat}]\label{chain_def}
    Given integers $k \ge 3$ and $\tau \ge 1$,
    a $\tau$-\emph{chain} $\mathcal{S}^{(k)}$ is a sequence of $\tau$ $k$-clauses $\langle C_1, \dots, C_{\tau} \rangle$ satisfies that $\forall i, j \in [\tau]$, $V(C_i) \cap V(C_j) = \emptyset$ if and only if $|i - j | > 1$.
\end{definition}
That is, a chain is a sequence of clauses with each of them must and only share variable with the adjacent clauses.
Therefore, a conjugate pair is a $2$-chain since there are only two clauses with the same variables.

The characteristic value of a chain is the solution $\lambda$ to the Linear Programming $\mathcal{LP}_k$ defined in Lemma~\ref{dls_upper_bound} with solution space containing all satisfying assignments of the chain.

\begin{lemma}[\cite{liu2018ksat}]\label{dls_upper_bound_ksat}
    Given $k$-instance $F$ and a set $\mathcal{I}$ of independent chains, \textsf{DLS} runs in time $T_{\text{DLS}} = \mathcal{O}((\frac{2(k-1)}{k})^{n'} \cdot \prod_{i}^{\chi} {\lambda_i}^{-\nu_i})$, where $n'$ is the number of variables not occurring in $\mathcal{I}$, $\chi$ is the number of different types of chains, $\lambda_i$ is the characteristic value of chain $\mathcal{S}_i$, and $\nu_i$ is number of chains in $\mathcal{I}$ with the same solution space to $\mathcal{S}_i$.
\end{lemma}

For conjugate pairs, it is easy to see that the solution space is just $A = \{0, 1\}^k \backslash \{0^k, 1^k\}$.
In our branching algorithm for NAE-$k$-SAT (Algorithm~\ref{br_k}), there is only one type of chain, which is the conjugate pair.
So we obtain:
\begin{itemize}
    \item $n' = n - k |\mathcal{P}|$.
    \item $\chi = 1$.
    \item $\nu_1 = |\mathcal{I}|$ is equal to the $|\mathcal{P}|$ defined in Lemma~\ref{dls_upper_bound}.
    \item $\lambda_1$ is equal to the $\lambda$ defined in Lemma~\ref{dls_upper_bound}.
\end{itemize}
Therefore Lemma~\ref{dls_upper_bound_ksat} immediately implies Lemma~\ref{dls_upper_bound}.

Note that in \cite{liu2018ksat}, the branching algorithm guarantees to find a subsequent clause sharing at most two variables with the current clause, therefore there is no conjugate pairs during its execution.

\section{Comparison results on MAX-NAE-$k$-SAT approximation}\label{sec_comp}

\begin{figure*}[h]
        \caption{The $x$-axis is the performance ratio $\delta$ and the $y$-axis is the base $c$ in the upper bound $c^n$.
        \emph{EPT(MAX-SAT)} \cite{DBLP:journals/tcs/EscoffierPT14} under our generalization is currently the fastest for MAX-$k$-SAT approximation, but it does not apply to MAX-NAE-$k$-SAT approximation (see \S{\ref{rw_apx}}).
        \emph{Hirsch} \cite{DBLP:journals/dam/Hirsch03} under our generalization is the fastest algorithm for MAX-NAE-$k$-SAT approximation (see \S{\ref{rw_subsection}}).
        Our tighter result outperforms the generalized \emph{Hirsch} on all $\delta$ and all $k$, and is even better than \emph{EPT(MAX-SAT)} when $\delta$ closes to $1$ and $k=3$.
        }
        \label{fig:four_figs}
        \centering
        \begin{minipage}[t]{0.410\textwidth}
            \centering
            \includegraphics[width=\textwidth]{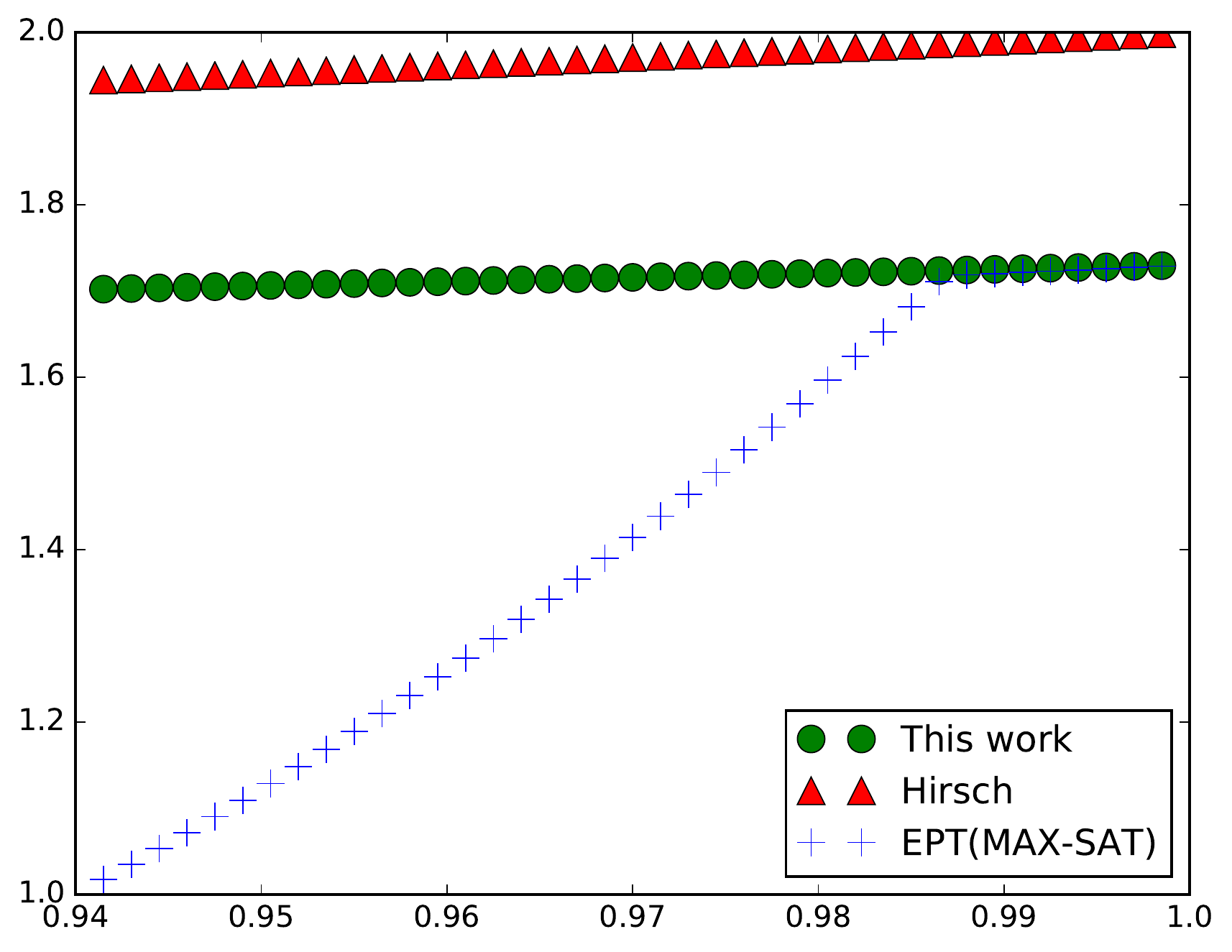}
            \subcaption{MAX-NAE-$2$-SAT}
            \label{fig:max-2-sat}
        \end{minipage}
        \hfill
        \begin{minipage}[t]{0.410\textwidth}
            \centering
            \includegraphics[width=\textwidth]{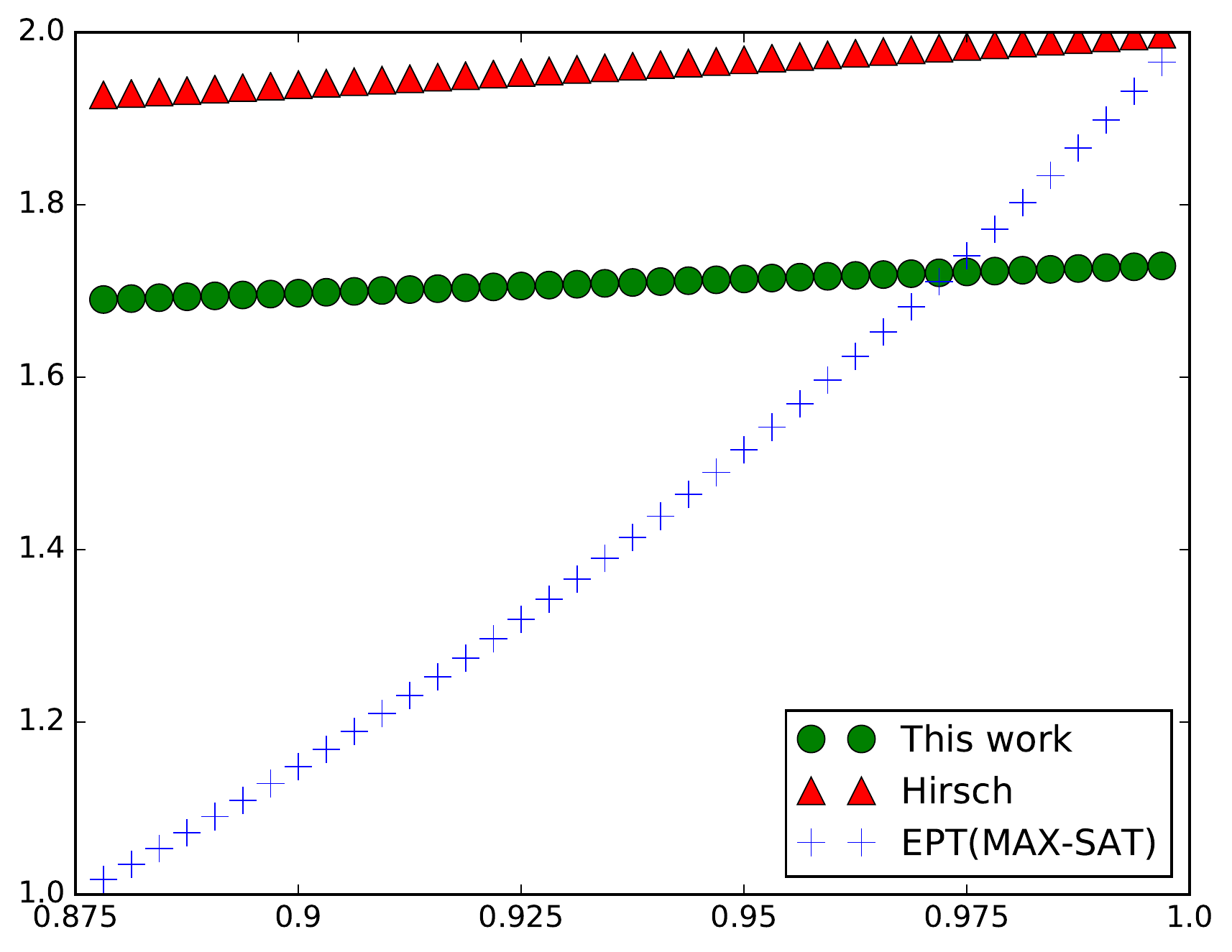}
            \subcaption{MAX-NAE-$3$-SAT}
            \label{fig:max-3-sat}
        \end{minipage}
        \begin{minipage}[t]{0.410\textwidth}
            \centering
            \includegraphics[width=\textwidth]{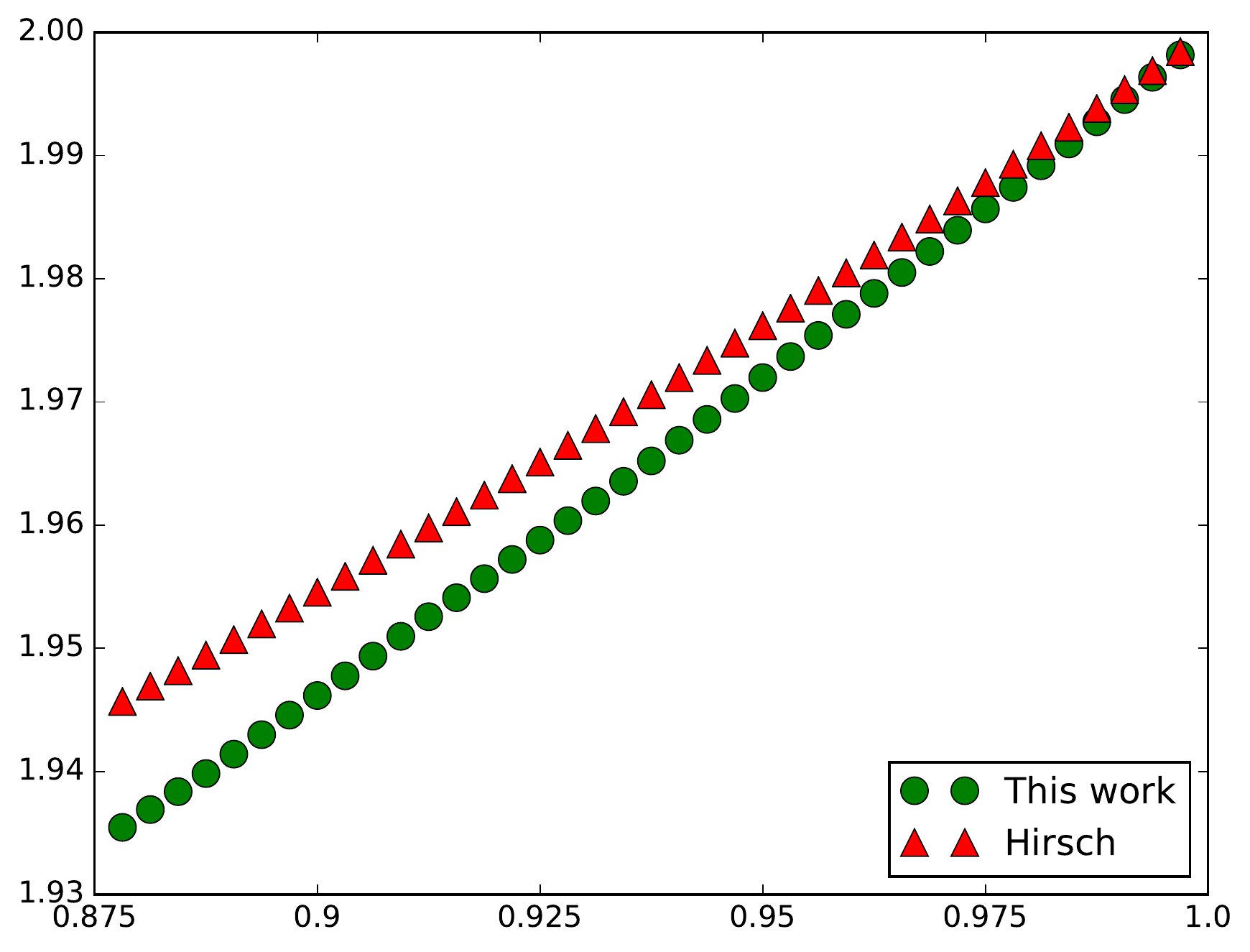}
            \subcaption{MAX-NAE-$4$-SAT}
            \label{fig:max-4-sat}
        \end{minipage}
        \hfill
        \begin{minipage}[t]{0.410\textwidth}
            \centering
            \includegraphics[width=\textwidth]{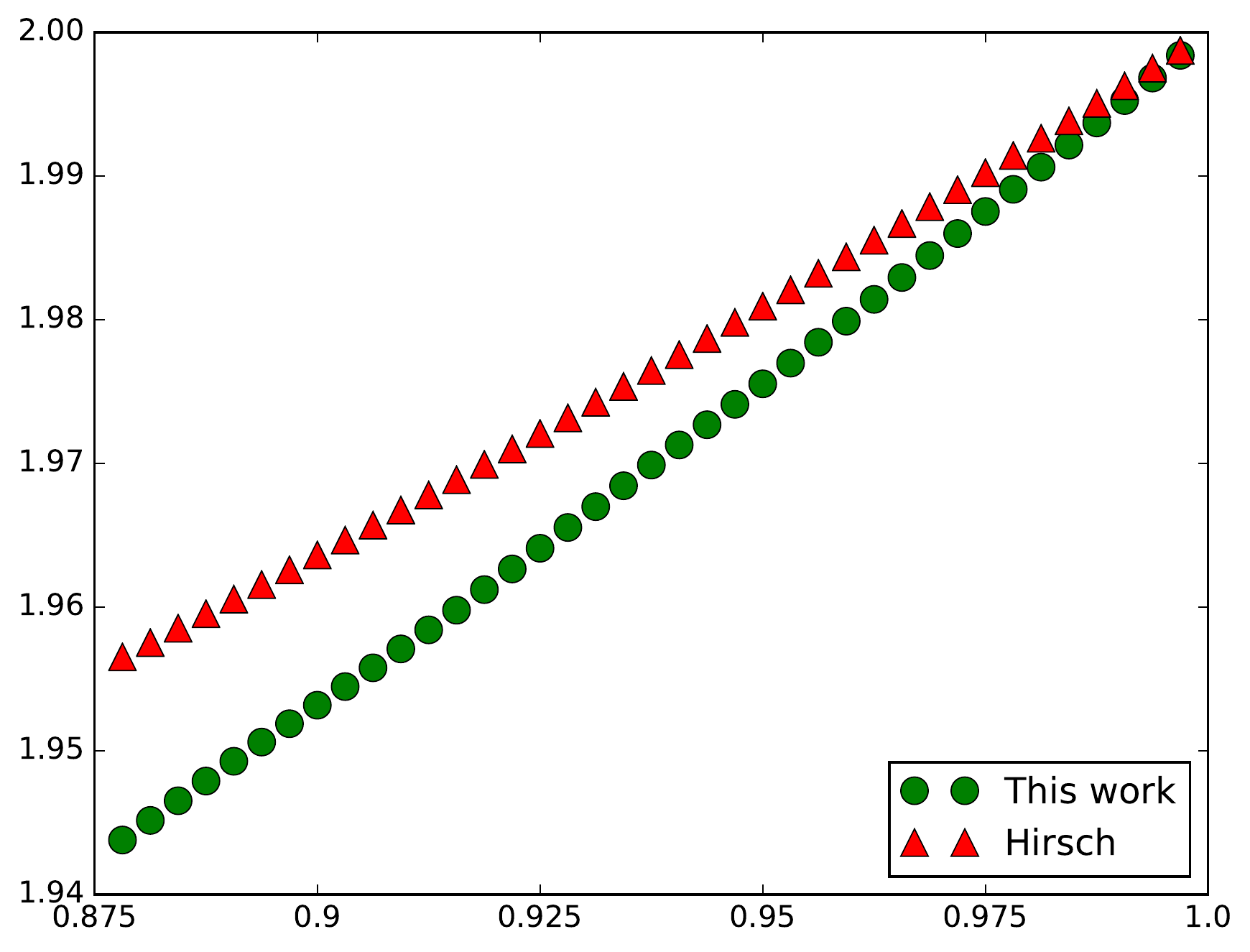}
            \subcaption{MAX-NAE-$5$-SAT}
            \label{fig:max-5-sat}
        \end{minipage}
    \end{figure*}

\end{document}